\documentclass[a4paper,aps,prl,showpacs,superscriptaddress,twocolumn]{revtex4-1}      
\usepackage[utf8]{inputenc}  
\usepackage[T1]{fontenc}     
\usepackage[british]{babel}  
\usepackage{lmodern}  
\usepackage[scaled=0.86]{berasans}  
\usepackage[scaled=1.03]{inconsolata} 
\usepackage[usenames,dvipsnames]{color} 
\usepackage[colorlinks,citecolor=blue,linkcolor=magenta,urlcolor=blue]{hyperref}  
\usepackage{graphicx} 
\usepackage[babel]{microtype}  
\usepackage{amsmath,amssymb,amsthm,bm,mathtools,amsfonts,mathrsfs,bbm} 
\usepackage{xspace}  
\usepackage{pgfplots}  
\DeclareMathOperator{\tr}{tr}
\DeclareMathOperator{\Tr}{Tr}

\newcommand{\bra}[1]{\left\langle #1 \right|}
\newcommand{\ket}[1]{\left| #1 \right\rangle}

\newcommand{\ketbra}[2]{\left|#1\middle\rangle\middle\langle#2\right|}

\newtheorem{lemma}{Lemma}
\newtheorem{definition}{Definition}

\newcommand{\ba}{\begin{eqnarray}}
\newcommand{\ea}{\end{eqnarray}}
\newcommand{\ban}{\begin{eqnarray*}}
\newcommand{\ean}{\end{eqnarray*}}

\begin{document}
\title{Algorithmic construction of local hidden variable models for entangled quantum states}

\author{Flavien Hirsch}
\affiliation{D\'epartement de Physique Th\'eorique, Universit\'e de Gen\`eve, 1211 Gen\`eve, Switzerland}

\author{Marco T\'ulio Quintino}
\affiliation{D\'epartement de Physique Th\'eorique, Universit\'e de Gen\`eve, 1211 Gen\`eve, Switzerland}

\author{Tam\'as V\'ertesi}
\affiliation{Institute for Nuclear Research, Hungarian Academy of Sciences,
H-4001 Debrecen, P.O. Box 51, Hungary}

\author{Matthew F. Pusey}
\affiliation{Perimeter Institute for Theoretical Physics, 31 Caroline Street North, Waterloo, ON N2L 2Y5, Canada}

\author{Nicolas Brunner}
\affiliation{D\'epartement de Physique Th\'eorique, Universit\'e de Gen\`eve, 1211 Gen\`eve, Switzerland}

\date{\today}  

\begin{abstract}
Constructing local hidden variable (LHV) models for entangled quantum states is challenging, as the model should reproduce quantum predictions for all possible local measurements. Here we present a simple method for building LHV models, applicable to general entangled states and considering continuous sets of measurements. 
This leads to a sequence of tests which, in the limit, fully capture the set of quantum states admitting a LHV model. Similar methods are developed for constructing local hidden state models. We illustrate the practical relevance of these methods with several examples, and discuss further applications.
\end{abstract}

\maketitle

Distant observers performing well-chosen local measurements on a shared entangled state can establish nonlocal correlations, as witnessed by the violation of a Bell inequality \cite{bell64,review}. Quantum nonlocality is among the most counter-intuitive features of quantum physics, and a key resource in quantum information processing \cite{DI,colbeck,pironio}.

Initially believed to be two different facets of the same phenomenon, entanglement and nonlocality are now recognized as fundamentally different. Notably, there exist entangled states which cannot give rise to nonlocality considering arbitrary (non-sequential) measurements. The correlations of such states---thus referred to as `local' entangled states---can be perfectly reproduced using a local hidden variable (LHV) model, i.e. using only shared classical resources. This was first demonstrated by Werner \cite{werner89}, who presented a class of entangled states which admit a LHV model for arbitrary projective measurements. This was later extended to more general POVMs \cite{barrett02}, and other classes of states \cite{almeida07,acin06,hirsch13,augusiak_review}. In particular, several works \cite{wiseman07,bowles14,sania14,quintino15} constructed local hidden state models (LHS), a special class of LHV model in which one party's hidden variable can be understood as a quantum state \cite{wiseman07}.

Constructing an LHV (or LHS) model for an entangled state is a challenging problem, since the model should reproduce the quantum statistics for a \emph{continuous} set of measurements, for instance all projective measurements. LHV (or LHS) models could be constructed for entangled states featuring a certain degree of symmetry \cite{augusiak_review}. Recently, a sufficient condition for a two-qubit state to admit a LHS was discussed \cite{joe15}. However, for general states, essentially nothing is known, due to the lack of appropriate techniques for discussing the problem.

Here we present a simple and efficient method for constructing LHV and LHS models, applicable to arbitrary entangled states and considering continuous sets of measurements. The main idea is to map the problem of finding a local model for an entangled state (a seemingly infinite problem) to a finite (hence tractable) problem, namely to find out whether the correlations resulting from a finite set of measurements on a different entangled state admit a local decomposition. We can define a sequence of tests for determining whether a given entangled state admits a LHV (or LHS) model, which is shown to converge in the limit, and thus to give a full characterization of the set of local entangled states (see Fig.1). The method can be efficiently implemented, and we construct LHV and LHS models for different classes of entangled states. In particular, we present LHS models for a non full-rank entangled state, and for a bound entangled state. We conclude by discussing further possible applications.

\begin{figure}[b!] \begin{center}
\includegraphics[width=0.85\columnwidth]{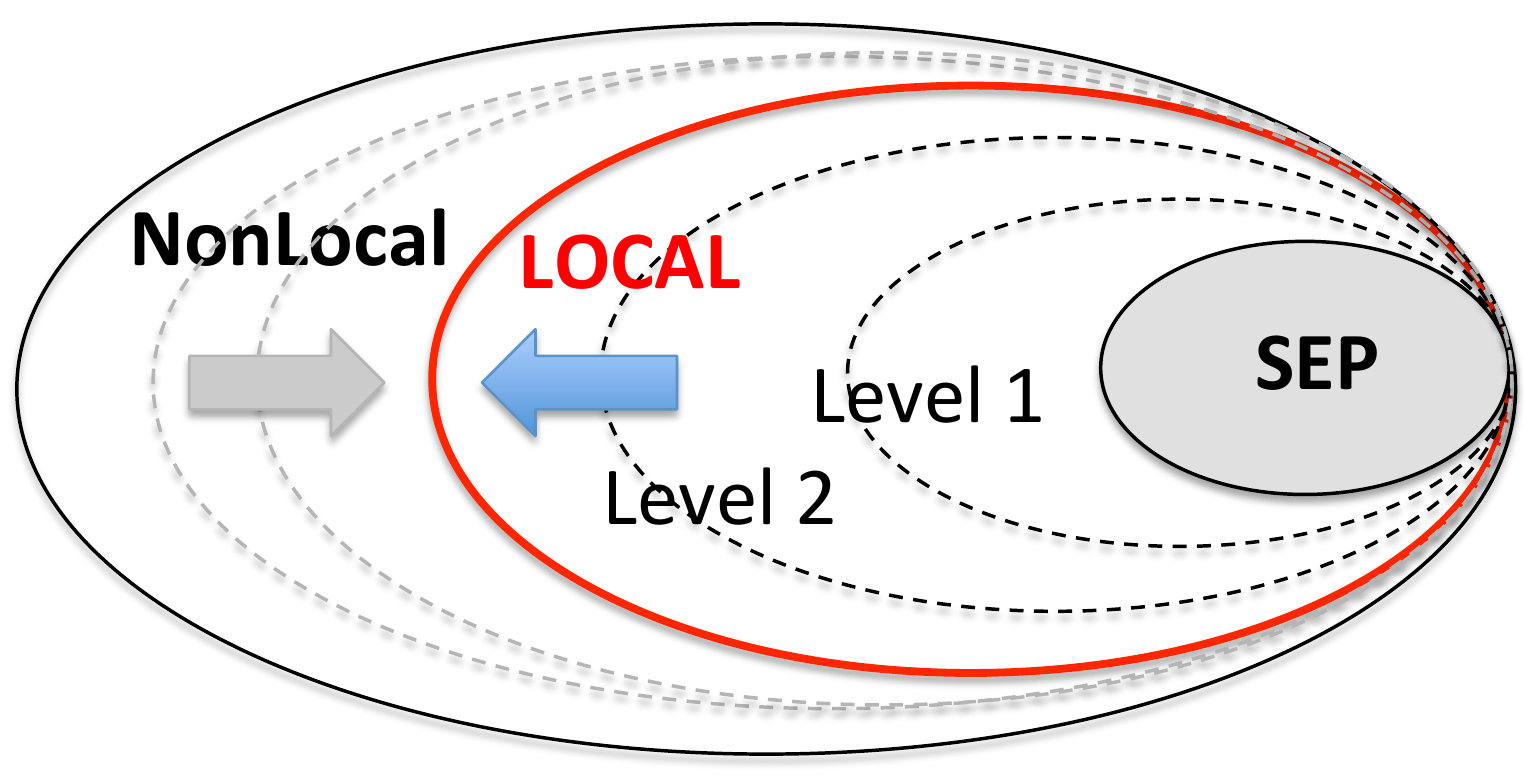}
\caption{A method for constructing LHV models for entangled states is discussed. This leads to a sequence of tests, which provide in each level a better approximation of the set of local states (red), a strict superset of the set of separable states (grey region). This is complementary to standard methods, based e.g. on Bell inequalities, which provide an approximation of the set of local states from outside.}
\end{center}
\end{figure}

\textit{Preliminaries} Consider Alice and Bob sharing an entangled quantum state $\rho$. Alice performs a set of measurements $\{M_{a|x} \}$ ($M_{a|x}\geq 0$ and $\sum_a M_{a|x} = \openone$), and Bob performs measurements $\{M_{b|y}\}$. The resulting statistics is given by
\begin{align} \label{pQ}
p(ab|xy) = \Tr ( M_{a|x} \otimes M_{b|y} \; \rho ).
\end{align}
The state $\rho$ is said to be local (for $\{M_{a|x} \}$ and $\{M_{b|y}\}$) if distribution \eqref{pQ} admits a Bell local decomposition:
\begin{align}
\label{LHV} p(ab|xy) &= \int \pi(\lambda) \; p_A (a|x,\lambda ) \; p_B (b|y,\lambda) \; d\lambda .
\end{align}
That is, the quantum statistics can be reproduced using a LHV model consisting of a shared local (hidden) variable $\lambda$, distributed with density $\pi(\lambda)$, and local response functions given by distributions $p_A (a|x,\lambda )$ and $p_B (b|y,\lambda)$. If a decomposition of the form \eqref{LHV} cannot be found, the distribution $p(ab|xy)$ violates (at least) one Bell inequality \cite{review}. In this case, we conclude that $\rho$ is nonlocal for the sets $\{M_{a|x} \}$ and $\{M_{b|y}\}$.

Another concept of interest is that of a LHS model, associated to quantum steering \cite{wiseman07}. Specifically, we say that $\rho$ is 'unsteerable' (from Alice to Bob) if
\begin{align}
\label{LHS} p(ab|xy) &= \int \pi(\lambda) \; p_A (a|x,\lambda ) \; \Tr ( M_{b|y} \sigma_\lambda) \; d\lambda .
\end{align}
That is, the quantum statistics can be reproduced by a LHS model, where $\sigma_\lambda$ denotes the local (hidden) quantum state and $p_A (a|x,\lambda )$ is Alice's response function. If such a decomposition cannot be found, $\rho$ is said to be 'steerable' for the set $\{M_{a|x}\}$; note that one would usually consider here a set of measurements $M_{b|y}$ that is tomographically complete, and thus focus the analysis on the set of conditional states of Bob's system
\ba
\sigma_{a|x} = \Tr_A (M_{a|x} \otimes \openone \; \rho),
\ea
referred to as an assemblage. Note also that any LHS model can be considered as an LHV model. The converse does not necessarily hold, as there exist entangled states which are steerable but nevertheless Bell local \cite{wiseman07,quintino15}.

The problem of testing the locality or unsteerability of a given entangled state $\rho$ for finite sets of measurements can be solved using existing methods, such as symmetric extensions for quantum states \cite{terhal}, linear and semi-definite programs (SDP)\cite{review,pusey13,skrzypczyk13}, and relaxing positivity \cite{virmani}. Implementable for small number of measurements, these methods become computationally demanding when increasing the number of measurements. Nevertheless, they are guaranteed to provide a solution in principle.

The situation is very different when considering continuous sets of measurements, e.g. the set of all projective measurements.
Here the methods for finite sets cannot be applied, not even in principle. One must then construct a LHV (or LHS) model explicitly, by exhibiting the distributions $\pi(\lambda) $ and response functions $p_A (a|x,\lambda ) $ and $p_B (b|y,\lambda)$. This was achieved for certain classes of entangled states, by exploiting their high level of symmetry. However, when considering general states, with less (or no) symmetry, following such an approach is extremely challenging.

In the present work, we follow a different path and present a general method for constructing LHV and LHS models for arbitrary states. The method can be efficiently implemented and will be illustrated with examples. Before presenting the main result we start with a simple example, providing the intuition behind our method.

\textit{Illustrative example}.-- Consider the class of two-qubit Werner states:
\ba \label{werner}
\rho_W(\alpha) = \alpha \ket{ \psi^{-}}\bra{ \psi^{-}}+(1-\alpha)\openone/4
\ea
where $\ket{ \psi^{-}}=(\ket{01}-\ket{10})/\sqrt{2}$ is the singlet state and $\openone/4$ is the two-qubit maximally mixed state. In the range $1/3 < \alpha \leq 1/2$, $\rho_W(\alpha)$ is entangled but unsteerable (hence local) for all projective measurements \cite{werner89}. Werner provided an explicit LHS model by exploiting the high symmetry of the state---$\rho_W(\alpha)$ is invariant under global rotations of the form $U \otimes U$. Here we illustrate the main idea behind our method by rederiving Werner's result, without invoking any symmetry argument.

Consider the set of 12 vectors $\hat{v}_x$ ($x=1,..,12$) on the Bloch sphere forming an icosahedron. This corresponds to a set $\mathcal{M}$ of 6 projective qubit measurements. By performing measurements in $\mathcal{M}$ on the Werner state, Alice prepares for Bob the assemblage
\ba \label{assemblageW}
\sigma_{\pm|x} = \Tr_A [\frac{\openone \pm \hat{v}_x \cdot \vec{\sigma}}{2} \otimes \openone \; \rho_W(\alpha)] ,
\ea
where $\vec{\sigma}$ denotes the vector of Pauli matrices. Using SDP techniques \cite{skrzypczyk13}, we find that this assemblage admits a LHS model for $\alpha \lesssim 0.54$.

This analysis can be extended to all projective measurements as follows. Consider qubit POVMs given by $M^\eta_{\pm | \hat{v}} =  ( \openone \pm \eta (\hat{v} \cdot \vec{\sigma} ) )/2$ with $0 < \eta \leq 1$. The corresponding Bloch vectors (with direction $\hat{v}$ and norm $\eta$) thus form a `shrunk' Bloch sphere of radius $\eta$. Choosing $\eta^* = \sqrt{(5+2\sqrt{5})/15} \approx 0.79$, we obtain a sphere which fits inside the icosahedron. Thus, any noisy measurement $M^{\eta^*}_{\pm | \hat{v}}$ can be expressed as a convex combination of measurements in $\mathcal{M}$ \cite{finite}. Since the assemblage \eqref{assemblageW} (resulting from measurements in $\mathcal{M}$) admits a LHS for $\alpha \lesssim 0.54$, we get that the assemblage resulting from any possible $M^\eta_{\pm | \hat{v}}$ with $\eta \leq \eta^*$ also admits a LHS model. Consequently, the statistics of arbitrary (but sufficiently noisy, i.e. $\eta \leq \eta^*$) measurements performed on the Werner state with $\alpha \lesssim 0.54$ can be simulated. Finally, notice that the statistics of noisy measurements on a given Werner state are equivalent to the statistics of projective measurements on a slightly more noisy Werner state:
\ba
\Tr_A [M^\eta_{\pm | \hat{v}} \otimes  \openone \rho_W(\alpha) ] = \Tr_A [M^1_{\pm | \hat{v}} \otimes  \openone \rho_W(\eta  \alpha) ]
\ea
Hence, states $\rho_W(\alpha)$ with $\alpha \lesssim 0.54 \eta^* \simeq 0.43$ admit a LHS model for all projective measurements. Note that by starting from a polyhedron with more (but nevertheless finitely many) vertices distributed (sufficiently evenly) over the sphere, the above procedure gives a LHS model for Werner states for $\alpha \rightarrow 1/2$ thus converging to Werner's model \cite{finite}. This is the optimal LHS model, since $\rho_W(\alpha)$ becomes steerable for $\alpha>1/2$ \cite{wiseman07}.

\textit{Constructing LHS models}.-- Based on the idea sketched above, we now present a general method for constructing LHS models for continuous sets of measurements, applicable to any entangled state. Formally, we will make use of the following result.

{\bf Lemma 1.} Consider a quantum state $\chi$ (of dimension $d \times d$), with reduced states $\chi_{A,B} = \Tr_{B,A}(\chi)$, and a finite set of measurements $ \{M_{a|x} \}$, such that the assemblage $ \sigma_{a|x} = \Tr_A(M_{a|x} \otimes \openone \chi)$ is unsteerable. Then the state
\ba
\rho= \eta \chi + (1-\eta) \xi_A \otimes \chi_B  \nonumber ,
\ea
where $\xi_A$ is an arbitrary density matrix (of dimension $d$), admits a LHS model for a continuous set of measurements
$\mathcal{M} $. The parameter $\eta$ corresponds to the 'shrinking factor' of $\mathcal{M} $ with respect to the finite set $ \{M_{a|x} \}$ (and given state $\xi_A$). Specifically, consider the continuous set of (shrunk) measurements
\ba
M_{a}^\eta = \eta M_{a} + (1-\eta) \Tr[\xi_A M_a] \openone_d
\ea
for any $M_a \in \mathcal{M} $. Then $\eta$ is the largest number such that all $M_{a}^\eta$ can be written as a convex combination of the elements of $ \{M_{a|x} \}$, i.e. $M_{a}^\eta = \sum_{x} p_x M_{a|x} $ with $\sum p_x=1$ and $p_x\geq 0$.

\begin{proof} The proof is based on the following relation
\ba
\Tr_A[ M_a^\eta \otimes \openone \chi] = \Tr_A[M_a \otimes \openone \rho].
\ea
Since $ \sigma_{a|x} $ is unsteerable, it follows that there exists a LHS model for $\chi$ and all (shrunk) measurements $M_a^\eta$. From the above equality, it follows that $\rho$ admits a LHS model for the continuous set of measurements $\mathcal{M} $.
\end{proof}

This allows us to get an explicit protocol for determining whether a given state $\rho$ admits a LHS model.

{\bf Protocol 1.} The problem is to determine if a target state $\rho$ admits a LHS model for a continuous set of measurements $\mathcal{M} $. Following Lemma 1, we start by picking a finite set $ \{M_{a|x} \}$ (with shrinking factor $\eta$) and a density matrix $\xi_A$. Next we solve the following SDP problem:
\ba \text{find  } & & q^* = \max  q  \\
\text{s.t.  }   & & Tr_A(M _{a|x} \otimes \mathbb{I} \, \chi) = \sum_\lambda \sigma_\lambda D_\lambda(a|x)  \,\,  \forall a,x,  \,\,  \sigma_\lambda \geq 0 \,\,  \forall \lambda \nonumber  \\ \nonumber
& &   \eta \chi + (1-\eta) \xi_A \otimes \chi_B = q \rho + (1-q) \frac{\openone}{d^2}
\ea
where the optimization variable are (i) the positive matrices $\sigma_\lambda$ and (ii) a $d \times d$ hermitian matrix $\chi$ \cite{footnote}. This SDP must be performed considering all possible deterministic strategies for Alice $D_\lambda(a|x)$, of which there are $N= (k_A)^{m_A}$ (where $m_A$ denotes the number of measurements of Alice and $k_A$ the number of outcomes); hence $\lambda = 1,...,N$. If the optimization returns a maximum of $q^*=1$, then $\rho$ admits a LHS model for all measurements in $\mathcal{M}$. If $q^*<1$ we have shown that $\rho' = q \rho + (1-q) \frac{\mathbb{I}}{d^2} $ with $q\leq q^*$ admits a LHS for $\mathcal{M}$.

The performance of the above protocol depends crucially on the choice of the set $ \{M_{a|x} \}$. It must be chosen in a rather uniform manner, over the continuous set $\mathcal{M}$, in order to get a shrinking factor as large as possible. Also, the ability of the protocol to detect a larger range of unsteerable states will improve when increasing the number of measurements contained in $ \{M_{a|x} \}$. Computing the shrinking factor is in general non-trivial, but we give a general procedure in the Appendix.

Based on Protocol 1, we can define a sequence of tests for unsteerability of a given target state $\rho$. In the first test, we consider a finite set $ \{M_{a|x} \}_1$, with shrinking factor $\eta_1$ and apply Protocol 1. We thus get a value of $q^*_1$. If $q^*_1=1$, we conclude that $\rho$ admits a LHS. On the other hand, if $q^*_1<1$, the test is inconclusive, and we must go to the second level. We construct now a new set  $ \{M_{a|x} \}_2$, which includes all measurements in $ \{M_{a|x} \}_1$ and additional ones. By adding sufficiently new measurements, we get a new shrinking factor $\eta_2> \eta_1$. Applying Protocol 1 again, we may get a value of $q^*_2 > q^*_1$ \cite{footnote2}. If $q^*_2=1$ we stop, otherwise we proceed to level 3, and so on.

Clearly, in each new test, the set of measurements considered provides a better approximation to $\mathcal{M}$. Moreover, the sequence of tests will in fact converge in the limit. Indeed, consider any state $\rho$ admitting a LHS model. Then, applying the method to $\rho$, we will be able to show that there is a state $\rho'$, arbitrarily close to $\rho$, which admits a LHS model. Specifically, for any $\epsilon>0$, the state $\rho = (1-\epsilon) \rho + \epsilon \frac{\openone}{d^2}$ will be detected by going to a sufficiently high level in the sequence of tests (see Appendix).

These ideas can be implemented on a standard computer for small sets of measurements $ \{M_{a|x} \}$. For larger sets, the implementation becomes demanding. Nevertheless, the method provides a definite answer in principle.

\textit{Constructing LHV models}.-- These ideas can also be adapted to the problem of constructing LHV models.

{\bf Lemma 2.}  Consider a state $\chi$ and finite sets of measurements $ \{M_{a|x} \}$, $\{M_{b|y} \}$ such that $ p(ab|xy) = \Tr(M_{a|x} \otimes M_{b|y} \chi)$ is local. Then the state
\ba
\rho &=& \eta \mu \rho + \eta (1-\mu) \rho_A \otimes \xi_B   \\ & &
+ \mu (1-\eta) \xi_A \otimes \rho_B +(1-\eta) (1-\mu) \xi_A \otimes \xi_B   \nonumber
\ea
admits a LHV model for the continuous sets of measurements $\mathcal{M_A} $ for Alice and $\mathcal{M_B}$ for Bob. Here $\xi_A$, $\xi_B$ are arbitrary density matrices (of dimension $d$), and $\eta$, $\mu$ denote the shrinking factors of $\mathcal{M_A} $, $\mathcal{M_B}$ with respect to $ \{M_{a|x} \}$, $\{M_{b|y} \}$.

The proof is a straightforward extension of that of Lemma 1. We now have the following protocol.

{\bf Protocol 2.} The problem is whether a target state $\rho$ admits a LHV model for measurements in $\mathcal{M_A}$ and $\mathcal{M_B}$. Following Lemma 2, we take finite sets $ \{M_{a|x} \}$ and $ \{M_{b|y} \}$ (with shrinking factors $\eta_A$, $\eta_B$) and density matrices $\xi_A$ and $\xi_B$. Then we solve the following linear problem:
\ba \text{find  } & & q^* = \max  q  \\
\text{s.t.  }   & & Tr(M _{a|x} \otimes M_{b|y} \chi) = \sum_\lambda p_\lambda D_\lambda(ab|xy)  \,\,  \forall a,b,x,y   \nonumber  \\ \nonumber    & & p_\lambda \geq 0 \,\,  \forall \lambda  \\
& &  q \rho + (1-q) \frac{\mathbb{I}}{d}    =  \eta \mu \chi + \eta (1-\mu) \chi_A \otimes \xi_B  \nonumber \\ \nonumber
& &  \quad \quad  \quad \quad  +\mu (1-\eta) \xi_A \otimes \chi_B +(1-\eta) (1-\mu) \xi_A \otimes \xi_B
\ea
where the optimization variable are (i) positive coefficients $p_\lambda$ and (ii) a $d \times d$ hermitian matrix $\chi$ \cite{footnote}. Given $m_A$ ($m_B$) measurements with $k_A$ ($k_B$) outcomes for Alice (Bob), one has $N = (k_A)^{m_A} (k_B)^{m_B}$ local deterministic strategies $D_\lambda(ab|xy)$, and $\lambda = 1,...,N $.

Again, this leads to a sequence of tests. 
In the first level, consider finite sets $ \{M_{a|x} \}_1$ and $ \{M_{b|y} \}_1$, with shrinking factors $\eta_1$ and $\mu_1$, and apply Protocol 2. If $q^*_1=1$, we conclude that $\rho$ admits a LHV model. If $q^*_1<1$, we proceed to the second level. We construct $ \{M_{a|x} \}_2$ and $ \{M_{b|y} \}_2$, including all measurements used in the first level plus additional ones. Hence we get better shrinking factors $\eta_2 \geq  \eta_1$ and $\mu_2 \geq \mu_1$. Applying Protocol 2, we may get a value of $q^*_2 > q^*_1$ \cite{footnote2}. If $q^*_2=1$ we stop, otherwise we go to level 3, and so on.

Here, the sequence will also converge here in the limit. Indeed, consider any local state $\rho$. There is $\rho'$, arbitrarily close to $\rho$, which the method will show to have a LHV model (see mental Material). Again, implementations on standard computers is possible for small sets $ \{M_{a|x} \}$ and $ \{M_{b|y} \}$.

\textit{Applications}.-- We now illustrate the practical relevance of the above methods, by constructing LHS and LHV model for classes of entangled states for which previous methods failed. A non-trivial issue is to obtain the shrinking factor for the sets of measurements that are used. For projective qubit measurements, this can be done efficiently by exploiting the Bloch sphere geometry (see Appendix). Hence we consider entangled states where (at least) one of the systems is a qubit, and focus primarily on projective measurements.

Consider first the class of two-qubit states:
\ba \label{MMM}
\rho(\alpha,\theta) = \alpha \ket{ \psi_{\theta}} \bra{ \psi_{\theta}} + (1-\alpha) \mathbb{I}_4/4
\ea
that is, partially entangled states $\ket{ \psi_{\theta}} = \cos\theta \ket{00} + \sin\theta \ket{11}$ mixed with white noise. The state is entangled for $\alpha > [1 + 2 \sin(2 \theta)]^{-1}$, via partial transposition \cite{peres}. Using Protocols 1 and 2 we find parameter ranges $\alpha,\theta$ where the state is unsteerable and local (see Fig.2); details in Appendices B and C. So far, relevant bounds for the locality of the above state were only given for $\theta = \pi/4$, i.e. for Werner states \eqref{werner}. In this case, we obtain an almost optimal LHS model ($\alpha \simeq 0.495$), and a LHV model which improves Werner's one ($\alpha \simeq 0.554$), but below the model of Ref. \cite{acin06} which achieved $\alpha \simeq 0.659$.

\begin{figure}[t!] \begin{center}
\includegraphics[width=\columnwidth]{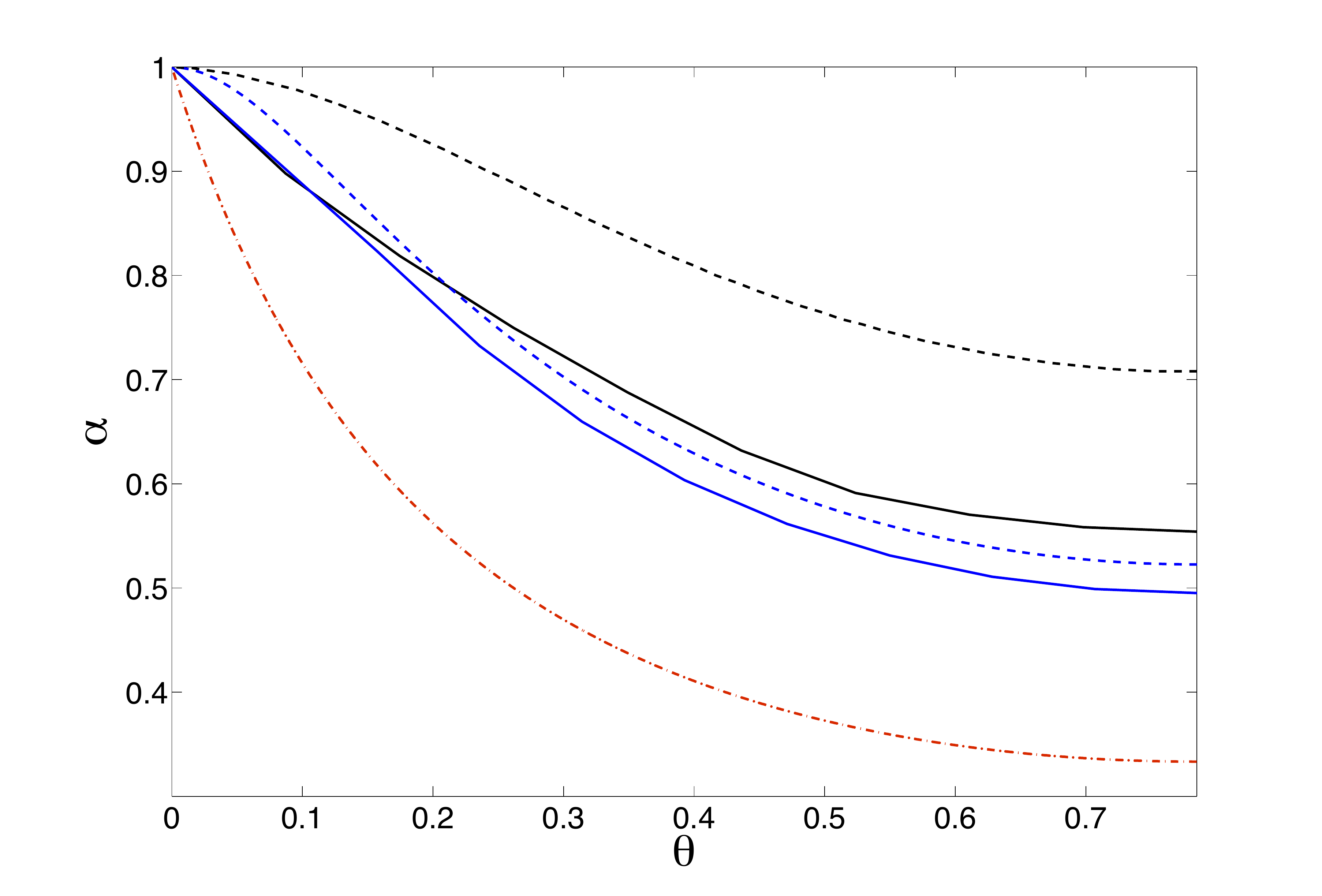}
\caption{The state $\rho(\alpha,\theta)$ of eq. \eqref{MMM} is entangled above the dash-dotted (red) line. Our method guarantees unsteerability below the solid blue line, while the state is steerable above the dashed blue line. Moreover, we can guarantee that the state is local below the solid black line, while it is nonlocal above the dashed black line.}
\end{center}
\end{figure}  	

Another application is to show that a non-full rank entangled state (i.e. on the boundary of the set of quantum states) can admit a LHS model. Specifically, we find that the state $ \rho = \sum_{k=1}^3 p_k \ket{\psi_k} \bra{\psi_k} $, where $p_1 = 0.4$, $p_2 = 0.05$, and
$\ket{\psi_1}=  \cos\theta \ket{00} + \sin\theta \ket{11}$, $\ket{ \psi_2} = \sin\theta \ket{00} - \cos\theta \ket{11}$ and $\ket{ \psi_3} =  \ket{10}$, where $\theta = 10^{-4} \pi$, admits a LHS model.

Next we discuss higher dimensional states, of the form
\ba \label{qubit-qudit}
\rho(\alpha,d) = \alpha \ket{ \psi^-} \bra{ \psi^-} + (1-\alpha) \openone_{2} / 2 \otimes \openone_{d} / d
\ea
i.e. a two-qubit singlet state $\ket{ \psi^{-}}$ mixed with higher dimensional noise. The above state is entangled for $\alpha > (1 + d )^{-1}$ (via partial transposition). We obtain lower bounds on $\alpha$ (for $d\leq 5$) for the state to admit a LHS model; see Appendix.


Moreover, we found that the well-known bound entangled state of Ref. \cite{pawel}, of dimension $2 \times 4$, admits a LHS model; see Appendix. This complements recent results showing that bound entangled states can lead to steering \cite{moroder14} and Bell nonlocality \cite{vertesi14}.

These methods can also be applied to multipartite entangled states. In particular, we could reproduce the result of Ref. \cite{toth05}, constructing a LHV model for a genuine tripartite entangled state.

Finally, we also applied our method considering general POVMs on the two-qubit Werner state \eqref{werner}. In this case, we obtain a LHS model for visibility $\alpha \simeq 0.36> 1/3$, which shows that the method can be applied in practice for general POVMs (see Appendix).

\textit{Discussion}.-- We discussed a procedure for constructing LHS or LHV models, applicable to any entangled states. The method can be used iteratively, and converges in the limit. We illustrated its practical relevance. Moreover, all models we construct require only a finite amount of shared randomness \cite{finite}.

We believe these methods will find further applications. First, we note that a simplified version of our method was recently used to demonstrate the effect of post-quantum steering \cite{belen}. More generally, the method can be applied to systems of arbitrary dimension, considering POVMs, and multipartite systems. Here the main technical difficulty consists in obtaining good shrinking factors for sets of measurements beyond projective qubit ones. Any progress in this direction would be interesting.

\emph{Acknowledgements.} We thank Joe Bowles, Daniel Cavalcanti, Leonardo Guerini, Rafael Rabelo, Denis Rosset and Paul Skrzypczyk for discussions. We acknowledge financial support from the Swiss National Science Foundation (grant PP00P2\_138917 and Starting grant DIAQ). Research at Perimeter Institute is supported in part by the Government of Canada
through NSERC and by the Province of Ontario through MRI.

\emph{Note added.} In independent work, Cavalcanti and colleagues presented related and complementary results \cite{cavalcanti15}.



\begin{thebibliography}{10}
\bibitem{footnote} Note that $\chi$ does not need to be positive, but should only define a valid assemblage (Protocol 1), or a valid probability distribution (Protocol 2), for the considered sets of measurements $ \{M_{a|x} \}$ (and $ \{M_{b|y} \}$).
\bibitem{footnote2} Note that value of $q_2^*$ depends on both the shrinking factor $\eta_2$ and the visibility (of $\chi$) given the measurements $\{M_{a|x} \}$. Hence increasing the shrinking factor does not necessarily imply that $q_2^*> q_1^*$.
\end{thebibliography}


\section{Appendix A. Calculating shrinking factors}

\subsection{Projective qubit measurements}

To apply Protocols 1 and 2, one first chooses a finite set of measurements $M_k$ and calculates the shrinking factor $\eta$ of a continuous set $\mathcal{M}$ with respect to $M_k$ (and given state $\xi_A$). This step is in general non-trivial. For the case of two-qubit projective measurements and $\xi_A=\openone_2/2$, we now present a simple and efficient method.

Consider the following set of POVMs
\ba
\mathcal{M}^{\eta} = \{ M^{\eta} | M^{\eta} = \eta M + (1-\eta) \mathbb{I} \}
\ea
where $M$ is a two-qubit projective measurement and $\mathbb{I} = \{ \openone_2 / 2,\openone_2 / 2 \}$. Using the Bloch sphere representation we can write any $M$ as
\ba
M = \{ M_{+}, M_{-} \} , \; M_{\pm} = \frac{ \openone \pm \hat{v} \cdot \vec{\sigma} }{2}
\ea
where $\vec{\sigma} = \{ \sigma_x, \sigma_y, \sigma_z \}$ contains the Pauli matrices and $\hat{v}$ is a normalized Bloch vector. For an element $M^{\eta}$ of the set $\mathcal{M}^{\eta}$ we therefore have
\ba
M^{\eta}= \{ M^{\eta}_{+}, M^{\eta}_{-} \} , \; M^{\eta}_{\pm} = \frac{ \openone \pm \eta \hat{v} \cdot \vec{\sigma} }{2} .
\ea
That is, in the Bloch sphere representation, the set represents a sphere of radius $\eta$. Hence the shrinking factor of such a set with respect to a finite set of projectors $\{ M_x \}$, represented by Bloch vectors $\{ \hat{v}_x \}$, is simply the radius of the largest sphere that can fit inside the polyhedron generated by $\{ \hat{v}_x \}$. This radius can be computed with arbitrary precision for any polyhedron by characterising its facets, the radius of the inscribed sphere being then the distance from the center of the sphere to the closest facet.

In all applications we presented in the main text, we used highly symmetric polyhedra, which have the property that all facets are equidistant from the center. Namely, we have used the cube and the icosahedron, which have respective shrinking factors
\ba \nonumber
\eta_{\text{cube}} = \frac{1}{\sqrt{3}} \approx 0.577 \, , \,\, \eta_{\text{ico}} = \sqrt{ \frac{5+2\sqrt{5}}{15}} \approx 0.795.
\ea

Note that for more general measurements, i.e. two-qubit POMVs or higher dimensional measurements, the above method for computing the shrinking factor, or even to put lower bounds on it, cannot be directly applied, as there is no notion of Bloch sphere in those cases. Nevertheless, we provide an explicit method below.

\subsection{General measurements}


The method is applicable to any set of POVMs of arbitrary dimension and is essentially a generalization of the above case. Let us recall that the shrinking factor of $\mathcal{M} $ with respect to the finite set $ \{M_{a|x} \}$ (and given state $\xi_A$) is the largest $\eta$  such that any element of the continuous set of (shrunk) measurements $\mathcal{M}^{\eta}$ defined by
\ba
M_{a}^\eta = \eta M_{a} + (1-\eta) \Tr[\xi_A M_a] \openone_d
\ea
can be written as a convex combination of the elements of $ \{M_{a|x} \}$, i.e. $M_{a}^\eta = \sum_{x} p_x M_{a|x} $ ($\forall a$) with $\sum p_x=1$ and $p_x\geq 0$. Geometrically, this corresponds to the fact that $\mathcal{M}^{\eta}$ is inside the convex hull of the $ \{M_{a|x} \}$. This can be checked by using the facet-representation of the polytope defined by the $ \{M_{a|x} \}$. Indeed, one can use any linear parametrization of the $ \{M_{a|x} \}$ in order to write them as real vectors, meaning the polytope is described by the vertices $v_x \in \mathbb{R}^n$, and a point $p \in \mathbb{R}^n$ is inside the polytope if and only if 
\ba\label{dot}
(f_k,p) \leq b_k \; \forall k=1...N
\ea
where the $f_k \in \mathbb{R}^n$ are the facets of the polytope defined by the vertices $v_x$, with bounds $b_k$ (which represent their distances from the zero vector) and $(  ,  )$ is the dot product. To prove that a set $\mathcal{M}^{\eta}$ is contained inside the polytope  $ \{M_{a|x} \}$ one can therefore show that for each facet $f_k$ of this polytope there is no point $p \in \mathcal{M}^{\eta}$ such that $(f_k,p) > b_k$. Choosing $\mathcal{M} $ to be qudit POVMs, one can do it using SDPs. For convenience we can write the facets in a matrix form $F_k^a$, where $a$ denotes the POVM outcome. In this case, the dot product in \eqref{dot} between a facet and a POVM $\{N_a\}$ is given by $\sum_a \Tr[F_k^a N_a] $. Hence $\{N_a\}$ is inside the polytope if and only if $\sum_a \Tr[F_k^a N_a] \leq b_k$. Now if $\mathcal{M}$ is taken to be the set of $n$-outcome qudit POVMs consider the following SDP: 
\ba\label{SDP} \text{find  } & & \max  \sum_{a=1}^n \Tr[F_k^a N_a]   \\
\text{s.t.  }   & & N_a = \eta M_a + (1-\eta) \Tr[\xi_A M_a] \openone_d,  \,\,  \forall a=1..n  \,\,   \nonumber  \\ \nonumber
& & M_a \geq 0 \,\,  \forall a=1..n , \,\,\, \sum_{a=1}^n M_a = \openone_d \nonumber
\ea
where the optimization variables are the positive matrices $M_a$. One can run this SDP for each facet $F_k$ and if the result of the objective function is smaller than $b_k$ for every $k$ one has then ensured that the set $\mathcal{M}^{\eta}$ is inside the polytope $ \{M_{a|x} \}$, \textit{i.e.} $\eta$ is 'small enough', otherwise, $\mathcal{M}^{\eta}$ is not contained by the polytope $ \{M_{a|x} \}$, \textit{i.e.} $\eta$  is 'too big'. 

Note indeed that $\eta$ has to be fixed all along the process, but the method provides a way to obtain the shrinking factor with an arbitrary precision as it can be iteratively implemented: starting from any $0 < \eta_1 < 1$ if $\eta_1$ is small enough (meaning $\Tr[F_k^a N_a] \leq b_k$ for all $k$) it implies that $\eta_1 \leq \eta$, where $\eta$ is the actual shrinking factor. Then one can for instance choose $\eta_2 = (\eta_1 + 1)/2$, and if, say, $\eta_2$ is too big, it means that $\eta_2 \geq \eta$ and one can take $\eta_3=( \eta_1 + \eta_2)/2$, etc. until the required precision. The starting point $\eta_1$ can be chosen by a numerical estimate of $\eta$: for a fixed measurement $M \in \mathcal{M}$ it is an easy linear programming task to find the largest $p$ such that $M^p= p M_{a} + (1-p) \Tr[\xi_A M_a] \openone_d$ is inside the polytope defined by the $ \{M_{a|x} \}$. One can thus parametrize $M$ and using a minimization algorithm can try to find the $M$ that minimizes $p$, which will be an upper bound on $\eta$, but hopefully not too far from it.  

Note also that for any dimension $d$ the special case of two-outcome POVMs can be tackled without using SDP \eqref{SDP}. One has in fact to find the maximum of $\Tr[F_k^1 N_1]+\Tr[F_k^1 N_2]$ (where $N_a = \eta M_a + (1-\eta) \Tr[\xi_A M_a]$), which can be rewritten as $\Tr[(F_k^1 -F_k^2) N_1] + \Tr[F_k^1]$ (using $N_1+N_2=\openone$) and setting $F_k= F_k^1 -F_k^2$ we get
\ba
\Tr[F_k N_1] & &  = \Tr[F_k (\eta M_1 + (1-\eta) \Tr[\xi_A M_1] \openone_d)  \\ \nonumber
& & = \Tr[F_k^{\eta} M_1] 
\ea
where $F_k^{\eta} = \eta F_k + (1-\eta) \Tr[F_k] \xi_A$. The maximum is therefore the sum of positive eigenvalues of $F_k^{\eta}$, achieved by letting $M_1$ be the sum of the corresponding projectors. Notice in particular that if $\xi_A$ is the maximally mixed state the eigenvalues of $F_k^{\eta}$ are just the $\eta \lambda_l(F_k) + (1-\eta) \Tr[F_k]/d$, where the $\lambda_l(F_k)$ is the $l$-th eigenvalue of $F_k$.

\section{Appendix B. Illustration of iterative procedure}

In the main text we have discussed, as application of our methods, the following class of two-qubit states:
\ba \label{MMM2}
\rho(\alpha,\theta) = \alpha \ket{ \psi_{\theta}} \bra{ \psi_{\theta}} + (1-\alpha) \mathbb{I}_4/4
\ea
where $\ket{ \psi_{\theta}} = \cos(\theta) \ket{00} + \sin(\theta) \ket{11}$. In Fig.2 we presented the best results for LHS and LHV models. Below we give more details about how we obtained these curves, which also illustrate the different levels of the procedure.

In the first level, we use the icosahedron. That is, the finite set $ \{M_{a|x} \}_1$ (as well as $ \{M_{b|y} \}_1$ for LHV models) is given by Bloch vectors corresponding to the vertices of the icosahedron. We thus get shrinking factor $\eta_1 = \eta_\text{ico}$. Moving to the second level, we now add 10 new measurements (i.e. 20 vertices) corresponding to the geometrical dual of the icosahedron. We thus get a new set $ \{M_{a|x} \}_2$ (and $ \{M_{b|y} \}_2$ for LHV) with 16 measurements (32 vertices). This gives a better shrinking factor of $\eta_2 \approx 0.923$. In level 3, we proceed similarly (i.e. adding vertices of the dual) and obtain a polyhedron with 92 vertices, and a shrinking factor of $\eta_3 \approx 0.971$. In level 4, we get 272 vertices and $\eta_4 \approx 0.989$.

\begin{figure}[t!] \begin{center}
		\includegraphics[width=\columnwidth]{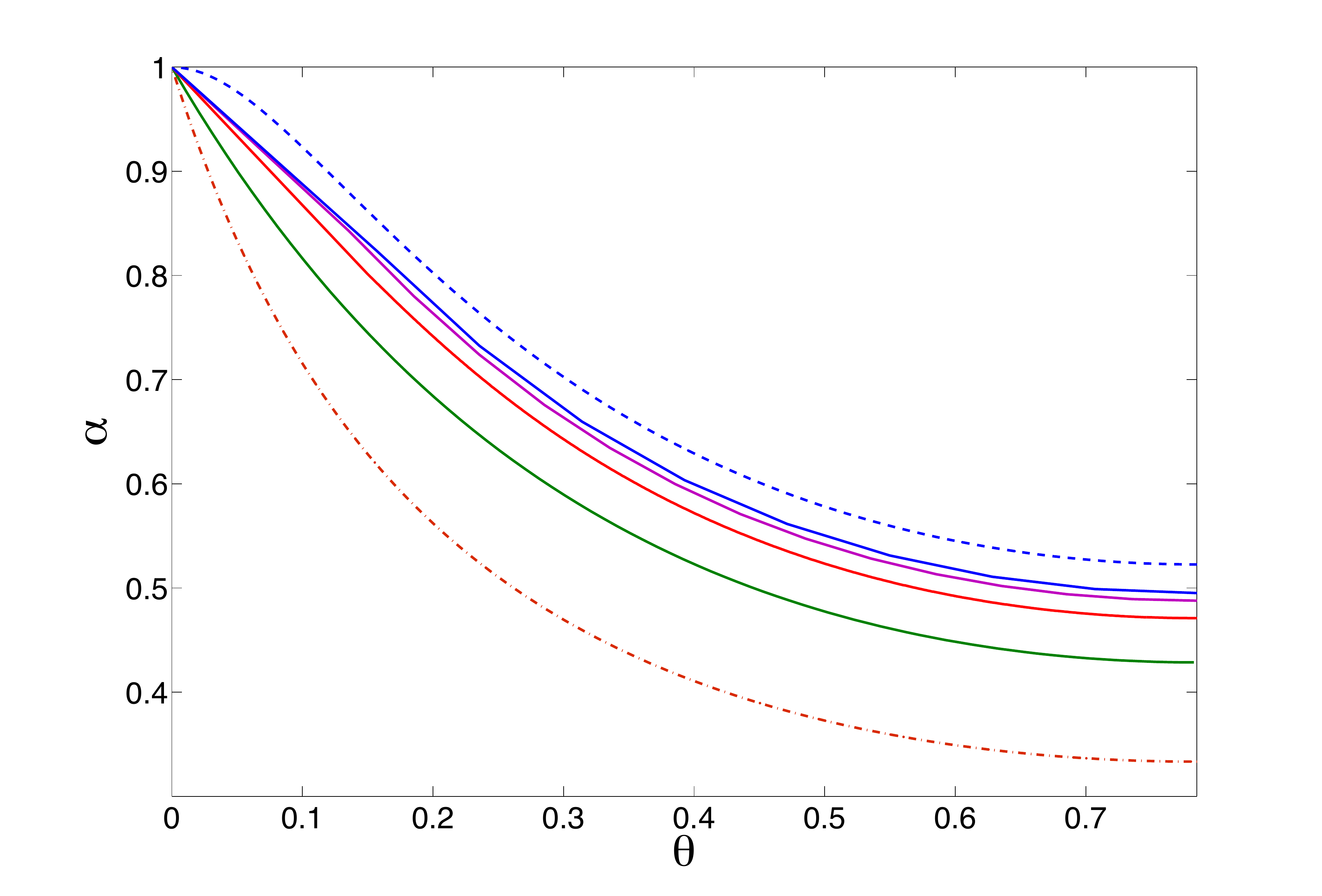}
		\caption{Results for unsteerability of states \eqref{MMM2}. The solid curves are obtained by using Protocol 1 at different levels in the sequence of tests; level 1 in green, level 2 in red, level 3 in purple, and level 4 in blue. Clearly, we get a better results in each level, and level 4 brings us relatively close to the (known) limit of steerability (dashed blue curve); obtained using 9 projective measurements and the SDP method of \cite{pusey13,skrzypczyk13}. We believe that the curve obtained in level 4 is close to the actual limit of steerability, the dashed blue curve being most probably suboptimal due to the small number of measurements considered. The separability limit is given by the dash-dotted red curve.}
	\end{center}
\end{figure}

\begin{figure}[t!] \begin{center}
		\includegraphics[width=\columnwidth]{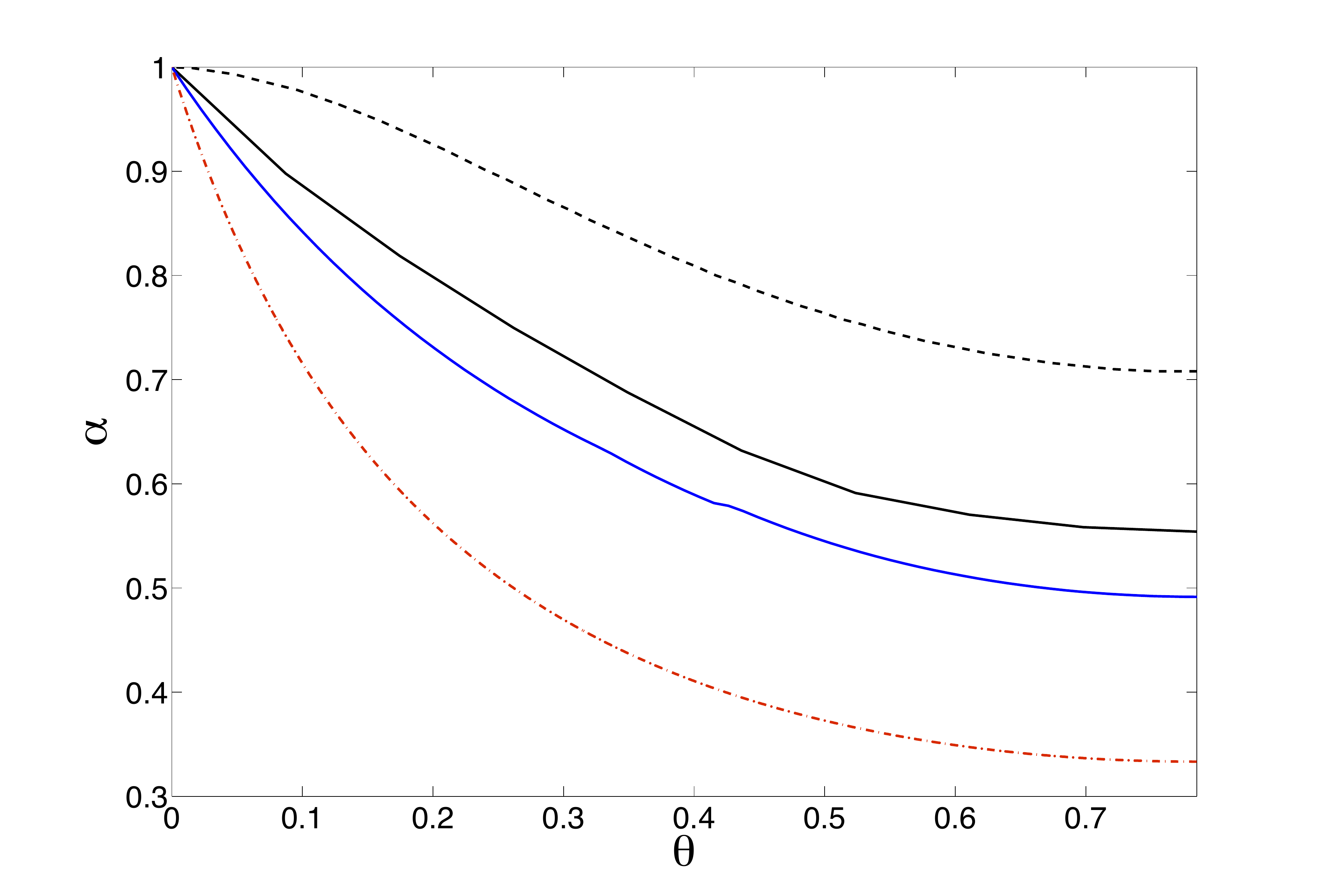}
		\caption{Results for the locality of states \eqref{MMM2}. The curves are obtained by using Protocol 2 at level 1 and 2 in the hierarchy. Again, we observe that level 2 considerably improves on level 1. We also give the limit of separability (dash-dotted red curve) and Bell nonlocality (dashed black curve), corresponding to CHSH violation via the Horodecki criterion \cite{horoCHSH}.}
	\end{center}
\end{figure}

The results we obtained for LHS (up to level 4) and LHV (up to level 2) models using the sequence of tests are presented on Figs. 3 and 4. Note that computation times grow rapidly when moving to higher levels. In particular, the main difficulty comes from the number of deterministic strategies to be considered, which grows exponentially as a function of the number of measurements. Nevertheless, one can obtain good bounds (even tight ones) by focussing on a subset of (well-chosen) deterministic strategies, which considerably speeds-up computation; details in Appendix C. Moreover, note that the orientation of the polyhedra must be taken into account. For instance, for states \eqref{MMM2} and small $\theta$, we observe that different orientations may lead to very different results.

\section{Appendix C. Choice of the deterministic strategies}

The time required to solve the SDP of Protocol 1, and the linear problem of Protocol 2, increases exponentially with the number of measurements $ \{M_{a|x} \}$ considered. This is due to the convex structure of the problem: Protocols 1 and 2 essentially find a decomposition of a point in terms of vertices of the local set, \textit{i.e.} the deterministic strategies (which are vertices of the set since they cannot be recovered as combination of other strategies). These are all the possibilities of attributing one outcome to a particular measurement. Considering $m$ measurements with $k$ outcomes, we have $k^m$ deterministic strategies.

The number of variables involved in the protocols is essentially given by the number of deterministic strategies, making the problem quickly infeasible as $m$ grows. On the other hand, taking more measurements allows for larger shrinking factors $\eta$ and eventually to detect more states. It is however possible to run the method using relatively large $m$ by using the following trick. Instead of considering all deterministic strategies, we focus on a subset of them. The bound we obtain might be suboptimal, but will nevertheless hold. Indeed, choosing appropriately the subset of deterministic points will be crucial.

Let us illustrate this method for the case of qubit projective measurements, and for the construction of LHS models via Protocol 1. Consider $m$ measurements, we thus have $2^m$ deterministic strategies. Up to $m = 16$, the problem is feasible considering all deterministic strategies (approx. 16 hours on a standard computer). However, in order to go to larger $m$, the number of strategies must be restricted. A simple and efficient technique here is the following: we select the strategies compatible with the response function of Werner's model \cite{werner89}. Specifically, for a Bloch vector $\hat{\lambda}$ and a  measurement direction $\hat{v}$ Alice outputs $\pm 1$ with probability:
\ba
p(\pm|\hat{v},\hat{\lambda}) = \frac{1 \pm \text{sign}(\hat{v} \cdot \hat{\lambda} ) }{2} .
\ea
That is, Alice outputs $+1$ for measurement directions which are in the half sphere around $\hat{\lambda}$, and $-1$ otherwise, as in Werner's original model \cite{werner89}. Here, given a set of $m$ measurements, with vectors $ \hat{v}_1,\hat{v}_2,...,\hat{v}_m $, we consider only deterministic strategies compatible with the above response function. For instance, three vectors which are not in the same half sphere cannot all give the same outcome. Hence, many deterministic strategies can be eliminated. In fact, the problem becomes now feasible on a standard computer up to $m \lesssim 200$, leading to shrinking factors of $\eta \simeq 0.99$.

\section{Appendix D. LHS model for bound entangled state and higher dimensional states}

In \cite{pawel}, the first (and possibly simplest) example of a bound entangled state ($2 \times 4$) is presented. First consider
\ba
\ket{\psi_i} = \frac{1}{\sqrt{2}} ( \ket{0, i} + \ket{1, i+1} )
\ea
for $i=0, 1, 2,$ and
\ba
\sigma_{insep}=\frac{2}{7} \sum^{3}_{i=1} \ketbra{\psi_i}{\psi_i} + \frac{1}{7} \ketbra{02}{02},
\ea
which is entangled (via the partial transpose criterion). Finally the family of states we are interested in are:
\ba
\sigma_b = \frac{7b}{7b + 1} \sigma_{insep} + \frac{1}{7b + 1} \ketbra{\phi_b}{\phi_b}.
\ea
where $\ket{\phi_{b}} = \ket{1} \otimes (\sqrt{\frac{1+b}{2}} \ket{0} + \sqrt{\frac{1-b}{2}} \ket{2}), 0 \leq b \leq 1.$
The state $\sigma_b$ is bound entangled for $ 0 < b < 1$. Using again Protocol 1 with the icosahedron, we find that $\sigma_b$ admits a LHS model for the whole range $b \in [0,1]$. This shows that a bound entangled state can admit a LHS model.

Next we discuss higher dimensional states, of the form
\ba \label{qubit-qudit}
\rho(\alpha,d) = \alpha \ket{ \psi^-} \bra{ \psi^-} + (1-\alpha) \openone_{2} / 2 \otimes \openone_{d} / d
\ea
where $\ket{ \psi^{-}}$ is the two-qubit singlet state, and $\openone_{d}$ the maximally mixed state of dimension $d$. Using the partial transposition criterion we get that this state is entangled if $\alpha > (1 + d )^{-1}$. We obtain lower bounds on $\alpha$ (for $d\leq 5$) for the state to admit a LHS model; see Table 1.

\begin{table}[t!] \label{table1}
	\begin{tabular}{| c || c|c|c|c|}
		\hline
		d  & 2 & 3 & 4 & 5   \\
		\hline \hline
		$\alpha_{ENT}$ & 0.33 & 0.25 & 0.20 & 0.16 \\
		$\alpha_{LHS}$ & 0.49 & 0.38 & 0.32 & 0.28 \\
		\hline
	\end{tabular} \caption{Bounds on the steerability of the states \eqref{qubit-qudit}, as a function of dimension $d$. The state is entangled for $\alpha \geq \alpha_{ENT}$, and unsteerable for $\alpha \leq \alpha_{LHS}$. }
\end{table}

\section{Appendix E. Local models for POVMs}

Using the technique mentioned above we can compute the shrinking factors of the set of two-qubit POVMs with respect to a finite set $ \{M_{a|x} \}$. We need only to consider the set of 4-outcome POVMs since extremal POVMs of dimension $d$ have at most $d^2$ outcomes \cite{dariano05}. Once again we can choose the icosahedron, \textit{ie} the set of projective measurements the directions of which are the vertices of the three-dimensional icosahedron. More precisely we consider all relabellings of $ \{ P_+ , P_- , 0 , 0 \}$ for $P_+$ being a projector onto a vertex of the icosahedron and $P_-$ onto the opposite direction. In addition we consider the four relabellings of the trivial measurement $ \{\openone_2  , 0 , 0 , 0 \}$, which comes for free as it cannot help to violate any steering or Bell inequalities and consequently does not even need to be inputed in Protocol 1 or 2 . The set thus have 76 elements, but we need to take into account only 6 of them when running the Protocol, corresponding to the vertices in the upper half sphere of the icosahedron.

With this choice we find that $\eta \approx 0.673$ and we can run again Protocol 1 and 2 for multi-qubit states, obtaining this time local models for all POVMs. For instance, applying it to the two qubit Werner state: 
\ba \label{werner}
\rho_W(\alpha) = \alpha \ket{ \psi^{-}}\bra{ \psi^{-}}+(1-\alpha)\openone/4
\ea
we find a LHS model up to up to $\alpha=0.54 \times 0.673=0.363 $ for all POVMs.

\section{Appendix F. Convergence of the sequence of tests}

The goal of this section is to prove the convergence of our sequence of tests. More concretely, for any state $\rho$ that is unsteerable from Alice to Bob, our method certifies, in finitely many steps, that the state $\rho^{\epsilon}= (1-\epsilon) \rho + \epsilon \frac{\openone}{d^2}$ for any $\epsilon>0$ is unsteerable from Alice to Bob. Similarly for states which admit a LHV model. Since the proof of convergence for the case of LHS and LHV models are essentially the same, we focus on the former.

For clarity, we break the proof in parts and the central idea can be explained in two steps: (i) showing that for any given finite set of measurements $\{M_{a\vert x} \}$ with shrinking factor $\eta$, the algorithm detects $\rho^{\epsilon}$ where $\epsilon$ becomes arbitrarily small when $\eta \rightarrow 1$, and (ii) it is possible to find a finite set of measurements for any $\eta<1$.

First, consider $\rho$ unsteerable from Alice to Bob (from now on we omit mentioning 'from Alice to Bob'), and $\{M_{a\vert x} \}$ a finite set of measurements with a shrinking factor $\eta$. Then the state $\rho^{\eta}= \eta \rho + (1-\eta) \frac{\openone}{d}\otimes \rho_B $, which is clearly unsteerable, will be detected by using Protocol 1 and choosing $\chi=\rho$ and $\xi_A=\openone/d$.

Next, note that if $\rho^\eta $ is unsteerable, then $\rho'=\eta' \rho + (1-\eta') \frac{\openone}{d^2} $ is also unsteerable, where $\eta'=\eta/(d(1-\eta)+\eta)$. This is straightforward as $\rho'$ can be expressed as convex mixture of $\rho^\eta$ and a separable state:
\ba
\rho'  &=& \lambda \rho^\eta + (1-\lambda) \frac{\openone}{d}\otimes \rho_B^\perp \\ \nonumber
\rho_B^\perp &=& \frac{1-\lambda \eta}{1-\lambda} \frac{\openone}{d} - \frac{\lambda (1-\eta)}{1-\lambda}\rho_B \,, \quad \lambda = \frac{\eta'}{\eta}.
\ea

Next we note that for any dimension $d$, $\eta' \rightarrow 1$ when $\eta \rightarrow 1$. Hence what is left to show is that the shrinking factor $\eta$ can be made arbitrarily close to one using a finite set of measurements. Intuitively this follows from the fact that any convex set can be approximated arbitrarily well by considering a finite number of (well-chosen) points on its boundary. Formally, this can be shown as follows. We first define a metric in the space of subsets of $\mathbb{R}^n$ that captures the intuition of maximal distance between two points of two different sets.
\begin{definition}
	The Hausdorff distance between two sets is defined as
	\ba
	d(A,B)=\max{\{\sup_{a\in A}d(a,B),\sup_{b\in B}d(b,A)\}}
	\ea
	where $d(x,A)=\inf_{a\in A} d(x,a)$ is the distance between a point $x\in \mathbb{R}^d$ and a set $A\subseteq \mathbb{R}^d$ for a given metric $d:\mathbb{R}^d\to \mathbb{R}^+$.
\end{definition}

Any compact convex set $\mathcal{K}\subseteq\mathbb{R}^n$ can be arbitrarily approximated by a family of $n$-vertex polytopes $\mathcal{P}_n$. This can be done by showing that the Hausdorff distance between these two sets becomes arbitrarily small when $n$ is large.

\begin{lemma}
	Let $\mathcal{K} \in \mathbb{R}^d$ be a compact convex set. There is a family of $n$-vertex polytope $\mathcal{P}_n$ that is contained  by $\mathcal{K}$  and $d(\mathcal{K},\mathcal{P}_n)\leq \frac{c(\mathcal{K})}{n^{2/(d-1)}}$ where $c(\mathcal{K})$ depends only on 		$\mathcal{K}$.
\end{lemma}

For a proof of that we suggest Refs.~\cite{bronstein75,bronstein08}.
\begin{lemma}
	Let $\mathcal{K}$ be a compact convex set in $\mathbb{R}^d$ and $\mathcal{K}_{in}$ a convex set strictly contained in $\mathcal{K}$. There is always a polytope $\mathcal{P}$ that contains $\mathcal{K}_{in}$ that is contained in $\mathcal{K}$.
	
\end{lemma}
\begin{proof}
	
	Let $d_{min}$ be the minimum distance between  the interior set and the boundary of the exterior one, that is
	$d_{min}=\inf_{x\in \mathcal{K}_{in}}d(x,\partial \mathcal{K})>0$. From the above lemma, there exists a polytope $\mathcal{P}$ that is contained by $\mathcal{K}$ and the maximal distance between any two points of $\mathcal{P}$ and $\mathcal{K}$ (Hausdorff distance) is strictly smaller than $d_{min}$, hence this polytope is necessarily in between these two convex sets.
\end{proof}

\begin{lemma} \label{family}
	Given any $\eta<1$, there exists a finite set $\{M_{a\vert x}\}$ that has a shrinking factor $\eta^*>\eta$.
\end{lemma}

\begin{proof}
	
	Since all quantum measurements can be decomposed as POVMs with rank-1 elements, it is sufficient to consider the convex hull of this set \cite{barrett02,haapasalo11}. Also, since extremal POVMs of dimension $d$ have at most $d^2$ outcomes, the convex hull of rank-1 measurements is compact for any fixed $d$ \cite{dariano05}.
	
	Define the set of noisy rank-1 measurements by applying the channel $M\mapsto\eta M + (1-\eta) \tr(M)\frac{I}{d}$ to each POVM element. For any $\eta<1$ all POVM elements are full-rank operators and hence strictly inside the set of rank-1 measurements \cite{dariano05}, which allows us to find a set $\{M_{a\vert x}\}$ with shrinking factor greater than $\eta$.
\end{proof}

\end{document}